\documentclass[12pt]{article}
\usepackage{e-jc}
\usepackage[colorlinks=true,citecolor=black,linkcolor=black,urlcolor=blue]{hyperref}
\usepackage{
	amsmath,
	amssymb,
	amsthm,
	array,
	caption,
	chngpage,
	float,
	graphicx,
	mathrsfs,
	multicol,
	multirow,
	numprint,
	rotating,
	sectsty,
	spreadtab,
	subcaption,
	tikz,
	url,
	verbatim,
	wrapfig
}
\usetikzlibrary{automata, positioning}

\newcommand{\eps}{\varepsilon}
\newcommand{\squarefree}{square-free}
\newcommand{\cubefree}{cube-free}
\newcommand{\Squarefree}{Square-free}
\newcommand{\Cubefree}{Cube-free}
\newcommand{\abs}[1]{\lvert#1\rvert}
\newcommand{\myUrl}[1]{
	\begin{center}
		{\small\url{#1}} 
	\end{center}
}

\theoremstyle{plain}
\newtheorem{thm}{Theorem}
\newtheorem{cor}[thm]{Corollary}
\newtheorem{lem}[thm]{Lemma}

\theoremstyle{definition}
\newtheorem{df}[thm]{Definition}
\newtheorem{exa}[thm]{Example}
\newtheorem{con}[thm]{Conjecture}

\theoremstyle{remark}
\newtheorem{rem}[thm]{Remark}

\DeclareMathOperator{\first}{first}
\DeclareMathOperator{\last}{last}
\title{
	Nondeterministic automatic complexity of overlap-free and almost {\squarefree} words
}

\author{Kayleigh K. Hyde\\
\small Schmid College of Science \& Technology\\[-0.8ex]
\small Chapman University\\[-0.8ex] 
\small Orange, California, U.S.A.\\
\small\tt khyde@chapman.edu\\
\and
Bj{\o}rn Kjos-Hanssen\thanks{This work was partially supported by a grant from the Simons Foundation (\#315188 to Bj\o rn Kjos-Hanssen).}\\
\small Department of Mathematics\\[-0.8ex]
\small University of Hawai\textquoteleft i at M\=anoa\\[-0.8ex]
\small Honolulu, Hawai\textquoteleft i, U.S.A.\\
\small\tt bjoernkh@hawaii.edu\\
}

\date{
	\dateline{November 24, 2014}{July 29, 2015}\\
	\small Mathematics Subject Classifications: 68R15, 68Q30
}

\begin{document}
	\maketitle{}
	\begin{abstract}
		Shallit and Wang studied deterministic automatic complexity of words. 
		They showed that the automatic Hausdorff dimension $I(\mathbf t)$ of the infinite Thue word satisfies $1/3\le I(\mathbf t)\le 1/2$.
		We improve that result by showing that $I(\mathbf t)= 1/2$.
		We prove that the nondeterministic automatic complexity $A_N(x)$ of a word $x$ of length $n$ is bounded by $b(n):=\lfloor n/2\rfloor + 1$.
		This enables us to define the complexity deficiency $D(x)=b(n)-A_N(x)$.
		If $x$ is square-free then $D(x)=0$. If $x$ is almost square-free in the sense of Fraenkel and Simpson,
		or if $x$ is a overlap-free binary word such as the infinite Thue--Morse word, then $D(x)\le 1$.
		On the other hand, there is no constant upper bound on $D$ for overlap-free words over a ternary alphabet,
		nor for cube-free words over a binary alphabet.

		The decision problem whether $D(x)\ge d$ for given $x$, $d$ belongs to $\mathrm{NP}\cap \mathrm{E}$.
	\end{abstract}

	\section{Introduction}
		The Kolmogorov complexity of a finite word $w$ is, roughly speaking,
		the length of the shortest description $w^*$ of $w$ in a fixed formal language.
		The description $w^*$ can be thought of as an optimally compressed version of $w$.
		Motivated by the non-computability of Kolmogorov complexity,
		Shallit and Wang \cite{MR1897300} studied a deterministic finite automaton analogue.
		A more recent approach is due to Calude, Salomaa, and Roblot \cite{Calude}.

		\begin{df}[Shallit and Wang \cite{MR1897300}]
			The \emph{automatic complexity} of a finite binary string \(x=x_1\cdots x_n\) is 
			the least number \(A_D(x)\) of states of a {deterministic finite automaton} \(M\) such that 
			\(x\) is the only string of length \(n\) in the language accepted by \(M\).
		\end{df}
		This complexity notion has the following two properties:
		\begin{enumerate}
			\item{} Most of the relevant automata end up
				having a ``dead state'' whose sole purpose is to absorb any irrelevant or
				unacceptable transitions.
			\item{} The complexity of a string can be changed by reversing it. For instance,
				\begin{equation}\label{eq1}
					A_D(011100) = 4 < 5 = A_D(001110).
				\end{equation}
			Equation \ref{eq1} was verified by a computer program; for the idea and a partial proof see Figure \ref{referee}.
			The anonymous referee of this article raised the question, which we have not been able to answer,
			whether the complexity of a string and its reverse can be arbitrarily far apart.
		\end{enumerate}

		\begin{figure}
			\begin{center}
				\begin{tikzpicture}[shorten >=1pt,node distance=1.5cm,on grid,auto]
					\node[state,initial]   (q_0)                {$q_0$};
					\node[state,accepting] (q_1) [right of=q_0] {$q_1$};
					\node[state]           (q_2) [right of=q_1] {$q_2$};
					\path[->] 
						(q_0) edge                node {$0$} (q_1)
						(q_1) edge                node {$1$} (q_2)
						(q_2) edge [loop above]   node {$1$} ()
						(q_2) edge [bend left=70] node {$0$} (q_0);
				\end{tikzpicture}
			\end{center}
			\caption{A witnessing automaton for the inequality $A_D(011100)\le 4$.
				All missing transitions go to a dead state $q_3$ which is not shown.
			}\label{referee}
		\end{figure}

		If we replace {deterministic finite automata} by {nondeterministic} ones, these properties disappear.
		The nondeterministic automatic complexity turns out to have other pleasant properties, such as a sharp linear upper bound.
		\paragraph{Technical ideas and results.} In this paper we develop some of the properties of nondeterministic automatic complexity.
		As a corollary we get a strengthening of a result of Shallit and Wang \cite{MR1897300}
		on the complexity of the infinite Thue--Morse word $\mathbf t$.
		Moreover, viewed through an NFA lens we can, in a sense, characterize the complexity of $\mathbf t$ exactly.
		A main technical idea is to extend \cite[Theorem 9]{MR1897300} which said that
		not only do squares, cubes and higher powers of a word have low complexity,
		but a word completely free of such powers must conversely have high complexity.
		The way we strengthen their results is by considering a variation
		on square-freeness and cube-freeness, \emph{overlap-freeness}.
		This notion also goes by the names of \emph{irreducibility} and \emph{strong cube-freeness} in the combinatorial literature.
		We also take up an idea from \cite[Theorem 8]{MR1897300} and use it to show that
		the natural decision problem associated with nondeterministic automatic complexity is in E = DTIME($2^{O(n)}$).
		This result is a theoretical complement to the practical fact that
		the nondeterministic automatic complexity can be computed reasonably quickly;
		to see it in action, for strings of length up to 23
		one can view automaton witnesses and check complexity using the following URL format
			\myUrl{http://math.hawaii.edu/wordpress/bjoern/complexity-of-110101101/}
		and check one's comprehension by playing a Complexity Guessing Game at
			\myUrl{http://math.hawaii.edu/wordpress/bjoern/software/web/complexity-guessing-game/}
		Let us now define our central notion and get started on developing its properties.
		Recall that a nondeterministic finite automaton (NFA) is assumed to have no $\varepsilon$-transitions, i.e.,
		it is not an $\mathrm{NFA}-\varepsilon$.
		\begin{df}\label{precise}
			The nondeterministic automatic complexity $A_N(w)$ of a word $w$ is the minimum number of states of an NFA $M$ accepting $w$
			such that there is only one accepting path in $M$ of length $\abs{w}$.
		\end{df}
		The minimum complexity $A_N(w)=1$ is only achieved by words of the form $a^n$ where $a$ is a single letter.
		\begin{thm}[Hyde \cite{Hyde}]\label{Hyde}
			The nondeterministic automatic complexity $A_N(x)$ of a string $x$ of length $n$ satisfies
			\[
				A_N(x) \le b(n):={\lfloor} n/2 {\rfloor} + 1\mathrm{.}
			\]
		\end{thm}
		\begin{proof}[Proof sketch.]
			If $x$ has odd length, it suffices to carefully consider the automaton in Figure \ref{fig1}.
			If $x$ has even length, a slightly modified automaton can be used.
		\end{proof}
		\begin{figure}[h]
			\begin{tikzpicture}[shorten >=1pt,node distance=1.5cm,on grid,auto]
				\node[state,initial, accepting] (q_1)   {$q_1$}; 
				\node[state] (q_2)     [right=of q_1   ] {$q_2$}; 
				\node[state] (q_3)     [right=of q_2   ] {$q_3$}; 
				\node[state] (q_4)     [right=of q_3   ] {$q_4$};
				\node        (q_dots)  [right=of q_4   ] {$\dots$};
				\node[state] (q_m)     [right=of q_dots] {$q_m$};
				\node[state] (q_{m+1}) [right=of q_m   ] {$q_{m+1}$}; 
				\path[->] 
					(q_1)     edge [bend left]  node           {$x_1$}     (q_2)
					(q_2)     edge [bend left]  node           {$x_2$}     (q_3)
					(q_3)     edge [bend left]  node           {$x_3$}     (q_4)
					(q_4)     edge [bend left]  node [pos=.45] {$x_4$}     (q_dots)
					(q_dots)  edge [bend left]  node [pos=.6]  {$x_{m-1}$} (q_m)
					(q_m)     edge [bend left]  node [pos=.56] {$x_m$}     (q_{m+1})
					(q_{m+1}) edge [loop above] node           {$x_{m+1}$} ()
					(q_{m+1}) edge [bend left]  node [pos=.45] {$x_{m+2}$} (q_m)
					(q_m)     edge [bend left]  node [pos=.4]  {$x_{m+3}$} (q_dots)
					(q_dots)  edge [bend left]  node [pos=.6]  {$x_{n-3}$} (q_4)
					(q_4)     edge [bend left]  node           {$x_{n-2}$} (q_3)
					(q_3)     edge [bend left]  node           {$x_{n-1}$} (q_2)
					(q_2)     edge [bend left]  node           {$x_n$}     (q_1);
			\end{tikzpicture}
			\caption{
				A nondeterministic finite automaton that only accepts one string
				$x= x_1x_2x_3x_4 \cdots x_n$ of length $n = 2m + 1$.
			}
			\label{fig1}
		\end{figure}

		\begin{df}
			The \emph{complexity deficiency} of a word $x$ of length $n$ is
			\[
				D_n(x) = D(x) = b(n) - A_N(x).
			\]
		\end{df}
		The distribution of $A_N(w)$ for $w$ of length $n\le 23$ is given in Table \ref{deficiency}.
		The notion of deficiency is motivated by the experimental observation that about half of all strings have deficiency 0.
		\begin{sidewaystable}
			\centering
			\begin{spreadtab}{{tabular}{|r|r|r|r|r|r|r|r|r|r|r|r|r|}}
				\hline
				@$n\backslash k$ &1&2&3&4&5&6&7&8&9&10&11&12\\
				\hline
				23 & :={8388608-sum(c2:m2)} & 6 & 20 & 58 & 164 &  430 & 2540 & 14252 &  80962 & 442278 & 2160662 & 5687234\\
				22 & :={4194304-sum(c3:m3)} & 6 & 20 & 58 & 164 &  502 & 2846 & 16024 &  94732 & 451368 & 2089418 & 1539164\\
				21 & :={2097152-sum(c4:m4)} & 6 & 20 & 58 & 176 &  496 & 3168 & 18720 & 108042 & 504794 & 1461670 &\\
				20 & :={1048576-sum(c5:m5)} & 6 & 20 & 58 & 164 &  430 & 3814 & 23328 & 115896 & 529148 &  375710 &\\
				19 &  :={524288-sum(c6:m6)} & 6 & 20 & 58 & 164 &  582 & 4996 & 26542 & 140668 & 351250 & &\\
				18 &  :={262144-sum(c7:m7)} & 6 & 20 & 58 & 188 &  598 & 5692 & 29990 & 136024 &  89566 & &\\
				17 &  :={131072-sum(c8:m8)} & 6 & 20 & 58 & 200 &  514 & 7102 & 37042 &  86128 & &&\\
				16 &   :={65536-sum(c9:m9)} & 6 & 20 & 58 & 164 &  752 & 7738 & 34320 &  22476 & &&\\
				15 & :={32768-sum(c10:m10)} & 6 & 20 & 58 & 226 &  908 & 8530 & 23018 & &&&\\
				14 & :={16384-sum(c11:m11)} & 6 & 20 & 58 & 244 & 1270 & 9668 &  5116 & &&&\\
				13 &  :={8192-sum(c12:m12)} & 6 & 20 & 64 & 250 & 2076 & 5774 &       & &&&\\
				12 &  :={4096-sum(c13:m13)} & 6 & 20 & 58 & 282 & 2090 & 1638 &       & &&&\\
				11 &  :={2048-sum(c14:m14)} & 6 & 20 & 58 & 564 & 1398 &      &       & &&&\\
				10 &  :={1024-sum(c15:m15)} & 6 & 20 & 64 & 588 &  344 &      &       & &&&\\
				 9 &   :={512-sum(c16:m16)} & 6 & 20 & 78 & 406 &      &      &       & &&&\\
				 8 &   :={256-sum(c17:m17)} & 6 & 20 & 130&  98 &      &      &       & &&&\\
				 7 &   :={128-sum(c18:m18)} & 6 & 22 & 98 &     &      &      &       & &&&\\
				 6 &    :={64-sum(c19:m19)} & 6 & 26 & 30 &     &      &      &       & &&&\\
				 5 &    :={32-sum(c20:m20)} & 6 & 24 &    &     &      &      &       & &&&\\
				 4 &    :={16-sum(c21:m21)} & 6 &  8 &    &     &      &      &       & &&&\\
				 3 &     :={8-sum(c22:m22)} & 6 &    &    &     &      &      &       & &&&\\
				 2 &     :={4-sum(c23:m23)} & 2 &    &    &     &      &      &       & &&&\\
				 1 &     :={2-sum(c24:m24)} &   &    &    &     &      &      &       & &&&\\
				 0 & 1                      &   &    &    &     &      &      &       & &&&\\
				\hline
			\end{spreadtab}
			\caption{The number of strings of length $0\le n\le 23$ having nondeterministic automatic complexity $k$.}\label{deficiency}
		\end{sidewaystable}
	\section{Time complexity}
		\begin{df}
			Let \textsc{DEFICIENCY} be the following decision problem.

			\emph{Given a binary word $w$ and an integer $d\ge 0$, is $D(w)>d$?}
		\end{df}
		\subsection{NP}
			Theorem \ref{np} is not surprising; we do not know whether \textsc{DEFICIENCY} is $\mathrm{NP}$-complete.
			\begin{thm}\label{np}
				\textsc{DEFICIENCY} is in $\mathrm{NP}$.
			\end{thm}
			\begin{proof}
				Shallit and Wang \cite[Theorem 2]{MR1897300} showed that
				one can efficiently determine whether a given DFA uniquely accepts $w$ among string of length $\abs{w}$.
				Hyde \cite[Theorem 2.2]{Hyde} extended that result to NFAs, from which the result easily follows.
			\end{proof}
		\subsection{E}
			\begin{df}
				Suppose $M$ is an NFA with $q$ states that uniquely accepts a word $x$ of length $n$.
				Throughout this paper we may assume that $M$ contains no edges except those
				traversed on input $x$.
				Consider the \emph{almost unlabeled transition diagram} of $M$, which is
				a directed graph whose vertices are the states of $M$
				and whose edges correspond to transitions.
				Each edge is labeled with a 0 except for an edge entering the initial state as described below.

				We define the \emph{accepting path} $P$ for $x$ to be the sequence of $n+1$ edges traversed in this graph,
				where we include as first element an edge labeled with the empty string $\eps$ that enters the initial state $q_0$ of $M$.

				We define the \emph{abbreviated accepting path} $P'$ to be the sequence of edges obtained from $P$ by considering each edge
				in order and deleting it if it has previously been traversed.
			\end{df}
			\begin{lem}\label{types}
				Let $v$ be a vertex visited by an abbreviated accepting path $P'=(e_0, \ldots, e_t)$.
				Then $v$ is of one of the following five types.
				\begin{enumerate}
					\item In-degree 1 (edge $e_i$), out-degree 1 (edge $e_{i+1}$).
					\item In-degree 2 (edges $e_i$ and $e_j$ with $j>i$), out-degree 1 ($e_{i+1}$).
					\item In-degree 1 (edge $e_i$), out-degree 2 (edges $e_{i+1}$ and $e_j$, $j>i+1$).
					\item In-degree 2 (edges $e_i$ and $e_j$ with $j>i$), out-degree 2 ($e_{i+1}$ and $e_{j+1}$).
					\item In-degree 1 (edge $e_t$), out-degree 0.\footnote{This type was omitted by Shallit and Wang.}
				\end{enumerate}
			\end{lem}
			\begin{proof}
				The out-degree and in-degree of each vertex encountered along $P'$ are both $\le 2$,
				since failure of this would imply non-uniqueness of accepting path.
				Since all the edges of $M$ are included in $P$, the list includes all the possible in-degree, out-degree combinations.
				We can define $i$ by the rule that $e_i$ is the first edge in $P'$ entering $v$.
				Again, since all the edges of $M$ are included in $P$,
				$e_{i+1}$ must be one of the edges contributing to the out-degree of $v$, if any,
				and $e_j$ must also be as specified in the types.
			\end{proof}
			Lemma \ref{types} implies that Definition \ref{def:E} makes sense.
			\begin{df}\label{def:E}
				For $0\le i\le t+1$ and $0\le n\le t+1$ we let $E(i,n)$ be a string representing the edges $(e_i, \cdots, e_n)$.
				The meaning of the symbols is as follows: 0 represents an edge.
				A left bracket $[$ represents a vertex that is the target of a backedge.
				A right bracket $]$ represents a backedge.
				The symbol $+$ represents a vertex of out-degree 2.
				When $i>n$, we set $E(i,n)=\eps$.
				Next, assuming we have defined $E(j,m)$ for all $m$ and all $j>i$,
				we can define $E(i,n)$ by considering the type of the vertex reached by the edge $e_i$.
				Let $a_i\in\{0,\eps\}$ be the label of $e_i$.
				\begin{enumerate}
					\item{} $E(i,n) := a_i E(i+1,n)$.
					\item{} $E(i,n) := a_i [ E(i+1,j-1)] E(j+1,n)$.
					\item{} $E(i,n) := a_i + E(i+1,n)$.
					\item{} $E(i,n) := a_i [+E(i+1,j-1)] E(j+1,n)$.
					\item{} $E(i,n) := a_i E(i+1,n)$.
				\end{enumerate}
			\end{df}
			\begin{figure}
				\begin{subfigure}[b]{0.5\textwidth}
					\centering%
					\begin{tabular}{|c|c|}
						\hline{}
						$E(i,n)$&Computation\\
						\hline{}
						$E( 0, 12)$ & $\eps E(1,12)=E(1,12)$\\
						$E( 1, 12)$ & $a_1 {[E(2,11)]}_{a_{12}}E(13,12)$\\
						$E(13, 12)$ & $\eps$\\
						$E( 2, 11)$ & $a_2 E( 3, 11)$\\
						$E( 3, 11)$ & $a_3 E( 4, 11)$\\
						$E( 4, 11)$ & $a_4 E( 5, 11)$\\
						$E( 5, 11)$ & $a_5 {[E( 6, 10)]}_{a_{11}} E(12,11)$\\
						$E( 6, 10)$ & $a_6 E( 7, 10)$\\
						$E( 7, 10)$ & $a_7 + E( 8, 10)$\\
						$E( 8, 10)$ & $a_8 {[E(9, 9)]}_{a_{10}} E(11,10)$\\
						$E( 9,  9)$ & $a_9 + E(10,9) = a_9+$\\
						$E( 8, 10)$ & $a_8 {[a_9+]}_{a_{10}}$\\
						$E( 7, 10)$ & $a_7 + a_8 {[a_9+]}_{a_{10}}$\\
						$E( 6, 10)$ & $a_6 a_7 + a_8 {[a_9+]}_{a_{10}}$\\
						$E( 5, 11)$ & $a_5 {[a_6 a_7 + a_8 {[a_9+]}_{a_{10}}]}_{a_{11}}$\\
						\hline
					\end{tabular}
					\caption{
						The $+$ marks the place of a loopback.
					}
					\label{computation}
				\end{subfigure}
				\begin{subfigure}[b]{0.1\textwidth}
				\end{subfigure}
				\begin{subfigure}[b]{0.4\textwidth}
					\centering%
					\begin{tikzpicture}[shorten >=1pt,node distance=1.5cm,on grid,auto]
						\node[state,initial]   (q_0)                {$q_0$}; 
						\node[state]           (q_1) [below of=q_0] {$q_1$}; 
						\node[state]           (q_2) [below of=q_1] {$q_2$}; 
						\node[state]           (q_3) [below of=q_2] {$q_3$};
						\node[state,accepting] (q_4) [below of=q_3] {$q_4$};
						\node[state]           (q_5) [below of=q_4] {$q_5$};
						\node[state]           (q_6) [below of=q_5] {$q_6$};
						\node[state]           (q_7) [below of=q_6] {$q_7$};
						\node[state]           (q_8) [below of=q_7] {$q_8$};
						\node[state]           (q_9) [below of=q_8] {$q_9$};
						\path[->] 
							(q_0) edge              node {$0$} (q_1)
							(q_1) edge              node {$1$} (q_2)
							(q_2) edge              node {$0$} (q_3)
							(q_3) edge              node {$0$} (q_4)
							(q_4) edge              node {$0$} (q_5)
							(q_5) edge              node {$1$} (q_6)
							(q_6) edge              node {$1$} (q_7)
							(q_7) edge              node {$0$} (q_8)
							(q_8) edge [bend left]  node {$0$} (q_9)
							(q_9) edge [bend left]  node {$1$} (q_8)
							(q_9) edge [bend left=70]  node {$1$} (q_5)
							(q_7) edge [bend right] node {$1$} (q_1);
					\end{tikzpicture}
					\caption{
						Complexity witness for the string 0100011001010101111100, one of the 2,655,140 simple strings of length $n=22$.
					}
				\end{subfigure}
				\caption{
					The code is
						$E(0,12)={a_1}{[a_2a_3a_4a_5 {[a_6 a_7 + a_8 {[a_9+]}_{a_{10}}]}_{a_{11}}]}_{a_{12}}$
					where
					$(a_1, \ldots, a_{12}) = (0,1,0,0,0,1,1,0,0,1,1,1)$. In reduced form, 
						$E(0,12)=0[0000 [0 0 + 0 [0+]]]$.
				}\label{fig:Figure3}
			\end{figure}
			\begin{lem}\label{reconstructed}
				The abbreviated accepting path 
				can be reconstructed from $E(0,t)$.
			\end{lem}
			We do not include the proof of Lemma \ref{reconstructed};
			instead, Figure \ref{fig:Figure3} gives an example of an automaton and the computation of $E(0,t)$.
			\begin{lem}\label{easy}
				\[
					\abs{E(a,b)}\le 2(b-a+1).
				\]
			\end{lem}
			\begin{proof}[Proof of Lemma \ref{easy}.]
				The four rules are
				\begin{enumerate}
					\item{} $E(i,n)=a_i E(i+1,n)$
					\item{} $E(i,n)=a_i {[ E(i+1,j-1)]}_{a_j} E(j+1,n)$
					\item{} $E(i,n)=a_i+E(i+1,n)$
					\item{} $E(i,n)=a_i {[+E(i+1,j-1)]}_{a_j} E(j+1,n)$
				\end{enumerate}
				So either
				\[
					\abs{E(i,n)}\le 2+\abs{E(i+1,n)}
				\]
				or
				\[
					\abs{E(i,n)}\le 4+\abs{E(i+1,j-1)} + \abs{E(j+1,n)}.
				\]
				So if by induction hypothesis $\abs{E(a,b)}\le 2(b-a+1)$ then
				\[
					\abs{E(i,n)}\le 2+2(n-i-1+1) = 2(n-i+1)
				\]
				or
				\[
					\abs{E(i,n)}\le 4+2(j-1-i-1+1) + 2(n-j-1+1) = 2(n-i+1).\qedhere
				\]
			\end{proof}
			\begin{thm}\label{singly}
	 			\textsc{DEFICIENCY} is in E.
			\end{thm}
			\begin{proof}
				Let $w$ be a word of a length $n$, and let $d\ge 0$. To determine whether $D(w)>d$,
				we must determine whether there exists an NFA $M$ with at most $\lfloor\frac{n}{2}\rfloor - d$ states
				which accepts $w$, and accepts no other word of length $n$.
				Since there are \emph{prima facie} more than single-exponentially many automata to consider,
				we consider instead codes $E(0,t)$ as in Definition \ref{def:E}.
				By Lemma \ref{reconstructed} we can recover the abbreviated accepting path $P'$ and hence $M$ from such a code.
				The number of edges $t$ is bounded by the string length $n$, so by Lemma \ref{easy}
				\[
					\abs{E(0,t)}\le 2(t+1) \le 2(n+1);
				\]
				since there are four symbols this gives
				\[
					4^{2(n+1)}=O(16^n)
				\]
				codes to consider.
				Finally, to check whether a given $M$ accepts uniquely takes only polynomially many steps, as in Theorem \ref{np}.
			\end{proof}
			\begin{rem}
				The bound $16^n$ counts many automata that are not uniquely accepting; the actual number may be closer to $3^n$ based on
				computational evidence.
			\end{rem}

	\section{Powers and complexity}
		In this section we shall exhibit infinite words all of whose prefixes have complexity deficiency bounded by 1.
		We say that such a word has a hereditary deficiency bound of 1.
		\subsection{{\Squarefree} words}
			\begin{lem}\label{lyndon:schuetzenberger}
				Let $x$ and $y$ be strings over an arbitrary alphabet with $xy = yx$.
				Then there is a string $z$ and integers $k$ and $\ell$ such that $x = z^k$ and $y = z^\ell$.
			\end{lem}
			Lemma \ref{lyndon:schuetzenberger} is proved in Shallit \cite[Theorem 2.3.3]{Shallit:2008:SCF:1434864}
			and is originally due to Lyndon and Schuetzenberger \cite{MR0162838}.
			\begin{df}\label{factor}
				A word $x$ is a \emph{factor} in a word $y$ if $y=uxv$ for some words $u$ and $v$.
				In this case we also say that $y$ \emph{contains} $x$.
			\end{df}

			We will use the following simple strengthening from DFAs to NFAs of a fact used in \cite[Theorem 9]{MR1897300}.
			\begin{thm}\label{nfa9fact}
				If an NFA $M$ uniquely accepts $w$ of length $n$, and visits a state $p$ at least $k+1$ times, where $k\ge 2$,
				during its computation on input $w$,
				then $w$ contains a $k$th power.
			\end{thm}
			\begin{proof}
				Let $w=w_0 w_1 \cdots w_k w_{k+1}$ where
				\begin{itemize}
					\item $w_0$ is the portion of $w$ read before the first visit to the state $p$,
					\item $w_i$ is the portion of $w$ read between visits number $i$ and $i+1$ to the state $p$ for $1\le i\le k$, and
					\item $w_{k+1}$ is the portion of $w$ read after the last visit to the state $p$.
				\end{itemize}
				Thus $\abs{w_i}\ge 1$ for each $1\le i\le k$, but it is possible to have $\abs{w_0}=0$ ($\abs{w_{k+1}}=0$) since the
				initial (final) state of $M$'s on input $w$ computation may be $p$.

				For any permutation $\pi$ on ${1, \ldots, k}$, $M$ accepts $w_0w_{\pi(1)}\cdots w_{\pi(k)}w_{k+1}$.
				Let $1\le j\le k$ be such that $w_j$ has minimal length and let
				\[
					\hat w_j = w_1\cdots w_{j-1}w_{j+1}\cdots w_k.
				\]
				Then $M$ also accepts
				\[
					w_0w_j \hat w_j w_{k+1}\quad\mathrm{and}\quad w_0\hat w_j w_j w_{k+1}.
				\]
				By uniqueness,
				\[
					w_0w_j \hat w_j w_{k+1} = w = w_0\hat w_j w_j w_{k+1}
				\]
				and so
				\[
					w_j \hat w_j = \hat w_j w_j.
				\]
				By Lemma \ref{lyndon:schuetzenberger}, $w_j$ and $\hat w_j$ are both powers of a string $z$.
				Since $\abs{\hat w_j}\ge (k-1)\abs{w_j}$, $w_j\hat w_j$ is at least a $k$th power of $z$, so $w$ contains a $k$th power of $z$.
			\end{proof}
			\begin{thm}[Extended Pigeonhole Principle]\label{php}
				If $aq+d$ pigeons are placed in $q$ pigeonholes where $d>0$,
				then it cannot be the case that all pigeonholes have at most $a$ pigeons;
				in fact, either
				\begin{itemize}
					\item{} there is a pigeonhole with at least $a+d$ pigeons; or
					\item{} there is a pigeonhole with at least $a+d-1$ pigeons, and another with $a+1$ pigeons; or
					\item{} there is a pigeonhole with at least $a+d-2$ pigeons, and another with $a+2$ pigeons; or		
					\item{} there is a pigeonhole with at least $a+d-2$ pigeons, and two others with $a+1$ pigeons; or
					\item{} all pigeonholes have at most $a+d-3$ pigeons (which is impossible if $a+d-3 \le a$ and $d>0$).
				\end{itemize}
			\end{thm}
			\begin{proof}
				Consider the maximum number of pigeons in a pigeonhole $m$. If $m\ge a+d$ we are in Case 1.
				If $m=a+d-1$, we consider all the other pigeons and pigeonholes;
				there are then $q-1$ pigeonholes and $aq+d-(a+d-1)=a(q-1)+1$ pigeons.
				By the plain Pigeonhole Principle, there is a pigeonhole with at least $a+1$ pigeons.
				If $m=a+d-2$, we repeat the argument,
				consider the maximum number of pigeons in a pigeonhole other than a given one with the maximum number
				of pigeons.
			\end{proof}
			We next strengthen a particular case of \cite[Theorem 9]{MR1897300} to NFAs.
			\begin{thm}\label{nfa9square}
				A {\squarefree} word has deficiency 0.
			\end{thm}
			\begin{proof}
				Suppose $w$ is a word of length $n=2k$ or $n=2k+1$, of deficiency $d$.
				Then there is a witnessing automaton $M$ with $q = k + 1 - d$ states.
				Since $n + 1\ge 2k+ 1 = 2(k + 1 - d) + 2d - 1 = 2q + (2d-1)$, by the Extended Pigeonhole Principle (Theorem \ref{php}),
				there is a state $p$ which is visited $2+(2d-1)=3$ times $t_1 < t_2 < t_3$, during the $n + 1$ times of the computation of $M$
				on input $w$ (and is not visited at any other times in the interval $[t_1,t_3]$).
				By Theorem \ref{nfa9fact}, $w$ contains a square.
			\end{proof}
			\begin{cor}
				There exists an infinite word of hereditary deficiency 0.
			\end{cor}
			\begin{proof}
				There is an infinite {\squarefree} word over the alphabet $\{0,1,2\}$ as shown by Thue \cite{ThueTwo}.
				The result follows from Theorem \ref{nfa9square}.
			\end{proof}
		\subsection{{\Cubefree} and overlap-free words}
			\begin{df}
				For a word $u$, let $\first(u)$ and $\last(u)$ denote the first
				and last letters of $u$, respectively.
				An \emph{overlap} is a word of the form $uu\first(u)$ (or equivalently, $\last(u)uu$).
				A word $w$ is \emph{overlap-free} if it does not contain any overlaps.
			\end{df}
			\begin{df}[Thue-Morse morphism]
				Let $\{0,1\}^*$ denote the set of all finite binary words.
				A \emph{morphism} is a function $\nu:\{0,1\}^*\rightarrow\{0,1\}^*$
				which is a homomorphism with respect to concatenation, in the sense that
				\[
					\nu(xy) = \nu(x)\nu(y)
				\]
				for all $x,y\in\{0,1\}^*$.
				The \emph{Thue-Morse morphism} is the unique morphism $\mu$ satisfying
				\[
					\mu(0) = 01,\qquad \mu(1) = 10.
				\]
			\end{df}
			\begin{thm}[Shelton and Soni \cite{MR787496}]\label{shelton}
				Let $a\ge 0$.
				The words $\mu^a(00)$ and $\mu^a(001001)$ are overlap-free squares of lengths $2^{a+1}$, $3\cdot 2^{a+1}$, respectively\footnote{
					There is a minor typo in Shelton and Soni's paper (line 10 of page 98),
					equivalent to writing $\mu^a(001)$ instead of $\mu^a(001001)$.
				}.
			\end{thm}
			\begin{exa}[Examples of Theorem \ref{shelton}.]
				The following overlap-free squares exemplify the first few possible lengths, 2, 4, 6, 8 and 12:
				\[
					00,\quad \mu(00)=0101,\quad 001001,\quad \mu^2(00)=01100110,
				\]
				\[
					\quad \mu(001001)=010110010110.
				\]
			\end{exa}
			Theorem \ref{shelton} is used in the proof of the following result.
			\begin{thm}[Shelton and Soni \cite{MR787496}]\label{soni}
				Let $\ell$ be a positive integer. The following are equivalent.
				\begin{enumerate}
					\item There exists an overlap-free binary word $y$ and a word $x$ such that $y$ contains $xx$ and $\ell=\abs{xx}$.
					\item $\ell\in \{2^a: a\ge 1\}\cup\{3\cdot 2^a:a\ge 1\}$.
				\end{enumerate}
			\end{thm}
			\begin{lem}\label{cubeContainsCube}
				If a cube $www$ contains another cube $xxx$ then either
				$\abs{x}=\abs{w}$, or
				$xx\first(x)$ is contained in the first two consecutive occurrences of $w$, or
				$\last(x)xx$ is contained in the last two occurrences of $w$.
			\end{lem}
			\begin{proof}
				We prove the contrapositive. Suppose $xx\first(x)$ is not contained in the first two consecutive occurrences of $w$, and
				$\last(x)xx$ is not contained in the last two occurrences of $w$.
				Then the middle $\last(x)x\first(x)$ of the factor $xxx$ has $\last(w)w\first(w)$ as a factor, and hence $\abs{x}\ge\abs{w}$.
			\end{proof}
			\begin{thm}
				The deficiency of {\cubefree} binary words is unbounded.
			\end{thm}
			\begin{proof}
				Given $k$, we shall find a {\cubefree} word $x$ with $D(x)\ge k$.
				Pick a number $n$ such that $2^n\ge 2k+1$.
				Let $w := \mu^n(0)$, which is a word of length $\ell := 2^n$.
				By Theorem \ref{shelton}, $ww$ is overlap-free.
				Let $x=ww\hat w$ where $\hat w$ is the proper prefix of $w$ of length $\abs{w}-1$.
				By Lemma \ref{cubeContainsCube}, $x$ is {\cubefree}.
				The complexity of $x$ is at most $\abs{w}$ as
				we can just make one loop of length $w$, with code (Theorem \ref{singly})
				\[
					{[w_1\cdots w_{\ell-1}]}_{w_{\ell}}.
				\]
				And so
				\begin{align*}
					D(x)&\ge \left\lfloor \frac{\abs{x}}2\right\rfloor + 1 - \abs{w} \ge \frac{\abs{x}}2 - \abs{w}
				\\
					&= \frac{3\abs{w}-1}{2} - \abs{w}
					= \frac{\abs{w}}2 - \frac12 \ge k.\qedhere\\
				\end{align*}
			\end{proof}
		\subsection{Overlap-free words}
			\begin{thm}[Thue \cite{ThueTwo}]\label{thueStrong}
				The infinite Thue--Morse word
				\[
					\mathbf t = t_0t_1\cdots = 0110\, 1001\, 1001\, 0110\cdots
				\]
				given by
				\[
				 	b=\sum b_i2^i,\quad b_i\in\{0,1\}
					\quad\Longrightarrow\quad
					t_b=\sum b_i\mod 2,
				\]
				is overlap-free.
			\end{thm}
			\begin{lem}\label{gelfond}
				Fix $j$ and $k$ and let $t_x$ denote the $x$th bit of the Thue--Morse word. The function
				\[
					f(u)=t_{x(u)-1}\quad\mathrm{where}\quad x(u) = 3^{k - j}(3u + 2)
				\]
				is eventually nonconstant.
			\end{lem}
			\begin{proof}
				Gelfond \cite{MR0220693} showed that $\mathbf t$ has no infinite arithmetic progressions
				(see also Morgenbesser, Shallit, Stoll \cite{MR2793891}).
			\end{proof}

			\begin{lem}\label{triple9:12:17}
				For each $k\ge 1$ there is a sequence $x_{1,k}, \ldots, x_{k,k}$ of positive integers such that
				\[
					\sum_{i=1}^k a_i x_{i,k}
					=2\sum_{i=1}^k x_{i,k}
					 \quad\Longrightarrow\quad a_1 = \cdots = a_k = 2.
				\]
				Let $t_j$ denote bit $j$ of the infinite Thue--Morse word.
				Then we can ensure that
				\begin{enumerate}
					\item \label{1} $x_{i,k}+1<x_{i+1,k}$ and
					\item \label{2} $t_{x_{i,k}}\ne t_{x_{i+1,k}}$ for each $1\le i<k$.
				\end{enumerate}
			\end{lem}
			\begin{proof}
				Let
				\[
					x_{1,1}=1.
				\]
				Given $x_{1,k-1}, \ldots, x_{k-1,k-1}$, we let $x_{i,k}=3x_{i,k-1}$ for $i<k$ and $x_{k,k}=3u_{k-1}+2$
				for a sufficienctly large number $u_{k-1}$.
				Reducing the equation
				\[
				\sum_{i=1}^k a_i x_{i,k}
				=2\sum_{i=1}^k x_{i,k}
				\]
				modulo 3, we see that $a_k\equiv 2$ (mod $3$). If $a_k\ge 5$ then
				\[
					\sum_i a_i x_{i,k}\ge 5x_{k,k}=15u_{k-1}+10
				\]
				\[
					> 6\sum_{i<k}x_{i,k-1} + 6u_{k-1} +4 = 2\sum_{i<k} x_{i,k}+2(3u_{k-1}+2) = 2\sum_{i\le k}x_{i,k};
				\]
				provided
				\[
					3u_{k-1}+2> 2\sum_{i<k}x_{i,k-1}
				\]
				so we conclude $a_k=2$.
				Then we can cancel $a_k$, divide by three and reduce to the induction hypothesis.

				Thus our numbers are
				\[
				 	x_{1,2}=3,\quad x_{2,2}=3u_1+2,
				\]
				\[
				 	x_{1,3}=3^2,\quad x_{2,3}=3(3u_1+2),\quad x_{3,3}=3u_2+2
				\]
				and in general
				\[
					x_{j,k}=3^{k-j}(3u_{j-1}+2)
				\]
				To ensure (\ref{1}) we just take $u_{j-1}$ sufficiently big.
				To ensure (\ref{2}), we apply Lemma \ref{gelfond}.
			\end{proof}

			\begin{thm}\label{unbounded}
				The complexity deficiency of overlap-free words over an alphabet of size three is unbounded.
			\end{thm}
			\begin{proof}
				Let $d\ge 1$. We will show that there is a word $w$ of deficiency $D(w)\ge d$. Let $k=2d-1$.
				For each $1\le i\le k$ let $x_i = x_{k+1-i,k}$ where the $x_{j,k}$ are as in Lemma \ref{triple9:12:17}.
				Note that since $x_{i,k}+1<x_{i+1,k}$, we have $x_i>x_{i+1}+1$.
				Let
				\[
					w = {\left(2\prod_{i=1}^{x_1-1} t_i\right)}^2 t_{x_1}
					    {\left(2\prod_{i=1}^{x_2-1} t_i\right)}^2 t_{x_2}
					    {\left(2\prod_{i=1}^{x_3-1} t_i\right)}^2 \cdots t_{x_{k-1}}
					    {\left(2\prod_{i=1}^{x_k-1} t_i\right)}^2
				\]
				\[
					= \lambda_{1} t_{x_1} \lambda_{2} \cdots t_{x_{k-1}} \lambda_{k}
				\]
				where $\lambda_i = {(2\tau_i)}^2$, $\tau_i=\prod_{j=1}^{x_i-1} t_j$, and
				where $t_i$ is the $i$th bit of the infinite Thue--Morse word on $\{0,1\}$, which is overlap-free (Theorem \ref{thueStrong}).
				Let $M$ be the NFA with code (Theorem \ref{singly})
				\[
					[+0^{x_1-1}]0[+0^{x_2-1}]0\cdots 0*[+0^{x_k-1}],
				\]
				where $*$ indicates the accept state. Let $X=\sum_{i=1}^{k}x_i$.
				Then $M$ has $k-1+X$ edges but only $q=X$ states; and $w$ has length
				\[
					n = k-1+2X = 2(d - 1) + 2X,
				\]
				giving $n/2+1=d + X$.
				
				Suppose $v$ is a word accepted by $M$. Then $M$ on input $v$ goes through each
				loop of length $x_i$ some number of times $a_i\ge 0$, where
				\[
					k - 1 + \sum_{i=1}^k a_i x_i = \abs{v}.
				\]
				If additionally $\abs{v}=\abs{w}$, then by Lemma \ref{triple9:12:17} we have $a_1 = a_2 = \cdots = a_k$, and hence $v = w$.
				Thus
				\[
					D(w) \ge \lfloor n/2 + 1\rfloor - q = d+X-X =d.
				\]
				Below we prove that $w$ is overlap-free.
			\end{proof}
			\begin{proof}[Proof that the word $w$ in Theorem \ref{unbounded} is overlap-free.]
				Suppose a word $uu$ is contained in $w$.
			
				\textbf{Proof that the number of 2s in $uu$ is either 0 or 2.}
				Let $o_1, \ldots, o_{2a}$ denote the occurrences of 2s in $uu$ and suppose $a\ge 1$.
				Let $\delta_i = o_{i+1} - o_i$.
				Then the sequence $(\delta_1, \cdots, \delta_a)$ is an interval in the sequence
				\[
					(x_1 - 1, x_1, x_2 - 1, x_2, \ldots, x_{k - 1} - 1, x_{k - 1}, x_k - 1).
				\]
				Since $x_i > x_{i+1} + 1$, in particular $\abs{x_i - x_{i+1}}>1$ and so this sequence is injective, i.e.,
				no two entries are the same.
				But $(o_1, \cdots, o_a) = (o_{a+1} - \abs{u}, \cdots, o_{2a} - \abs{u})$.
				So $\delta_{a+1} = o_{a+2}-o_{a+1} = o_2 - o_1 = \delta_1$ which implies $a=1$.

				So either Case 1 or Case 2 below obtains.
				\textbf{Case 1: The number of 2s in $uu$ is zero.} Then certainly $uu\first(u)$ is not contained in $w$,
				since the infinite Thue--Morse word is overlap-free.
				\textbf{Case 2: The number of 2s in $uu$ is two.}
				Then we have one of the following two cases.
				\begin{enumerate}
					\item $uu$ is contained in a word of the form
						\[
							t_1 \cdots t_{x_i}\quad 2\quad t_1 \cdots t_{x_{i+1}-1}\quad 2\quad t_1 \cdots t_{x_{i+1}}.
						\]
						We guard against that by making sure that
						\begin{itemize}
							\item{} $t_{x_i}\ne t_{x_{i+1}-1}$ (Lemma \ref{triple9:12:17}) and
							\item{} $2\ne t_{x_{i+1}}$ (the Thue--Morse word uses only the letters $0$ and $1$)
						\end{itemize}
					\item $uu$ is contained in a word of the form
						\[
							t_1 \cdots t_{x_i-1}\quad 2\quad t_1 \cdots t_{x_i}\quad 2\quad t_1 \cdots t_{x_{i+1}-1}.
						\]
						Since $uu$ contains exactly two 2s and the $t_j$ are not 2s, it follows that
						$uu=a2b2c$ where $a$, $b$, $c$ are words over the
						binary alphabet $\{0,1\}$.
						Then $u=a2{b_1}={b_2}2c$ where $b={b_1}{b_2}$, so $a=b_2$, $c=b_1$ and so actually $u=a2c$ and $t_1 \cdots t_{x_i}=b=ca$.
						Here then $\abs{ca}=x_i$.
						If $\abs{a}\le 2$ then consequently
						\[
							x_i - 2 \le \abs{c} \le x_{i + 1} - 1,
						\]
						which contradicts $x_{i+1}<x_i-1$. If $\abs{a}\ge 2$ then we appeal to Lemma \ref{thue}.
				\end{enumerate}
			\end{proof}
			\begin{lem}\label{thue}
				$t_{x_i-2} t_{x_i-1}\, 2\, t_1 \cdots t_{x_i}2$ cannot be a factor of a square having only two 2s.
			\end{lem}
			\begin{proof}
				The Thue--Morse word is a concatenation of disjoint occurrences of the words 01 and 10.
				Each of these two words are of the form $z\overline z$ where $\overline z=1-z$.
				The idea now is that if $x_i$ is odd then say it ends in a lone 0 and 2, 02;
				then adding the next control bit will give something ending in 012, preventing a square.

				More precisely, since $t_1 \cdots t_{x_i-1}2$ having odd or even length ends in say
				$z\overline z2$ or $z\overline za2$ respectively,
				and then $t_1 \cdots t_{x_i-1}t_{x_i}2$ ends in $z\overline zb2$ or $z\overline z a\overline a2$, respectively;
				either way $t_1 \cdots t_{x_i-1}2$ and $t_1 \cdots t_{x_i-1}t_{x_i}2$ are incompatible.
			\end{proof}

			Definition \ref{precise} yields the following lemma.
			\begin{lem}\label{precisedef}
				Let $(q_0,q_1,\cdots)$ be the sequence of states visited by an NFA $M$ given an input word $w$.
				For any $t$, $t_1$, $t_2$, and $r_i$, $s_i$ with
				\[
					(p_1, r_1, \ldots, r_{t-2}, p_2) = (q_{t_1}, \ldots, q_{t_1+t})
				\]
				and
				\[
					(p_1, s_1, \ldots, s_{t-2}, p_2) = (q_{t_2}, \ldots, q_{t_2+t}),
				\]
				we have $r_i=s_i$ for each $i$.
			\end{lem}
			Note that in Lemma \ref{precisedef}, it may very well be that $t_1\ne t_2$.

			\begin{thm}\label{main}
				Overlap-free binary words have deficiency bound 1.
			\end{thm}
			\begin{proof}
				Suppose $w$ is a word satisfying $D(w)\ge 2$ and consider the sequence of states visited in a witnessing computation.
				As in the proof of Theorem \ref{almostMain}, either there is a state that is visited four times, and hence there is a cube in $w$,
				or there are three \emph{state cubes} (states that are visited three times each), and hence there are three squares in $w$.
				By Theorem \ref{soni}, a overlap-free binary word can only contain squares of length $2^a$, $3\cdot 2^a$,
				and hence can only contain powers $u^i$ where $\abs{u}$ is of the form $2^a$, $3\cdot 2^a$, and $i\le 2$.

				In particular, the length of one of the squares in the three state cubes must divide the length of another.
				So if these two state cubes are disjoint then the shorter one repeated
				can replace one occurrence of the longer one, contradicting Lemma \ref{precisedef}.

				So suppose we have two state cubes, at states $p_1$ and $p_2$, that overlap.
				At $p_1$ then we read consecutive words $ab$ that are powers $a=u^i$, $b=u^j$ of a word $u$, and since
				there are no cubes in $w$ it must be that $i=j=1$ and so actually $a=b$.
				And at $p_2$ we have words $c$, $d$ that are powers of a word $v$ and again the exponents are 1 and $c=d$.

				The overlap means that in one of the two excursions of the same length starting and ending at $p_1$,
				we visit $p_2$. By uniqueness of the accepting path
				we then visit $p_2$ in both of these excursions.
				If we suppose the state cubes are chosen to be of minimal length then we only visit $p_2$ once in each excursion.
				If we write $a=rs$ where $r$ is the word read when going from $p_1$ to $p_2$, and $s$ is the word going from $p_2$ to $p_1$, then
				$c=sr$ and $w$ contains $rsrsr$. In particular, $w$ contains an overlap.
			\end{proof}
			\begin{rem}
				In computability theory, the effective Hausdorff dimension $\mathrm{dim}$ and effective packing dimension $\mathrm{Dim}$ of
				a single infinite binary sequence $\mathbf u$ are defined, and related to Kolmogorov complexity $C$.
				It is shown (see \cite[Theorem 13.3.4 and Corollary 13.11.12]{DH}) that
				\[
					\mathrm{dim}(\mathbf u) = \liminf_n\frac{C(u_1 \cdots u_n)}{n},\text{ and}
					\quad \mathrm{Dim}(\mathbf u) = \limsup_n\frac{C(u_1 \cdots u_n)}{n}.
				\]
				These results, together with the idea that automatic complexity is a miniaturization of Kolmogorov complexity,
				constitute our motivation for making Definitions \ref{daHd} and \ref{naHd} below.
			\end{rem}
			\begin{df}\label{daHd}
				For an infinite word $\mathbf u$ define the \emph{deterministic automatic Hausdorff dimension} of $\mathbf u$ by
				\[
					I(\mathbf u)=\liminf_{u\mathrm{\ prefix\ of\ }\mathbf u} \frac{A_D(u)}{\abs{u}}.
				\]
				and the \emph{deterministic automatic packing dimension} of $\mathbf u$ by
				\[
					S(\mathbf u)=\limsup_{u\mathrm{\ prefix\ of\ }\mathbf u} \frac{A_D(u)}{\abs{u}}.
				\]
			\end{df}
			The connection between effective dimension and automatic dimension is not merely by analogy, as Theorem \ref{christmas2014} shows.
			\begin{thm}\label{christmas2014}
				If $\mathbf x$ is an infinite word with $\mathrm{dim}(\mathbf x)>0$, then $I(\mathbf x)>0$.
			\end{thm}
			\begin{proof}
				This follows from the Kolmogorov complexity calculation in \cite[Theorem 9]{MR1897300}.
			\end{proof}
			For nondeterministic complexity, in light of Theorem \ref{Hyde} it is natural to make the following definition.
			\begin{df}\label{naHd}
				Define the \emph{nondeterministic automatic Hausdorff dimension} of $\mathbf u$ by
				\[
					I_N(\mathbf u)=2\cdot\liminf_{u\mathrm{\ prefix\ of\ }\mathbf u} \frac{A_N(u)}{\abs{u}}
				\]
				and define $S_N$ analogously.
			\end{df}
			\begin{thm}[Shallit and Wang's Theorem 18]\label{toStrengthen}
				$\frac13\le I(\mathbf t)\le S(\mathbf t)\le\frac23$.
			\end{thm}
			We are now ready to strengthen Theorem \ref{toStrengthen}.
			\begin{thm}
				$I(\mathbf{t})= \frac12$, and $I_N(\mathbf t)=S_N(\mathbf t)=1$.
			\end{thm}
			\begin{proof}
				The inequality $I(\mathbf t)\ge \frac12$ and the fact that $I_N(\mathbf t)=S_N(\mathbf t)=1$ follow from the observation that
				the proof of Theorem \ref{main} applies equally for deterministic complexity.
				The inequality $I(\mathbf t)\le \frac12$ was already implicit in the proof of \cite[Theorem 18]{MR1897300}.
				Let $T(m)=t_0 \cdots t_{m-1}$.
				In the table they give, with $m=2^{2n+1}$, we read off the inequality $A_D(T(m)) \le m + 3 - 2^{2n} = \frac{m}2 +3$.
			\end{proof}
		\subsection{Almost {\squarefree} words}
			\begin{df}[Fraenkel and Simpson \cite{MR1309124}]
				A word whose square factors all belong to the set $\{00, 11, 0101\}$ is called \emph{almost {\squarefree}}.
			\end{df}

			\begin{thm}\label{almostMain}
				A word that is almost {\squarefree} has a deficiency bound of 1.
			\end{thm}
			\begin{proof}
				It is easy to verify for words of length at most 3.
				Suppose now $w$ has length at least 4.
				Suppose $w$ is a word of a length $n \in \{2k, 2k+1\}$ where $k\ge 2$, with deficiency at least 2.
				Then there are $q = k - 1\ge 1$ states occupied at $n + 1$ times.
				So $n+1 \in \{2k+1, 2k+2\} = \{2q+3, 2q+4\}$ times.
				There are at least $2q + 3$ times and only $q$ states, so
				by the Extended Pigeonhole Principle (Theorem \ref{php}), we are in one of the following
				Cases 1--3.
				\begin{itemize}
					\item{}
						Case 1. There is at least one state that is visited at least 5 times.
						Then by Theorem \ref{nfa9fact}, \textbf{$w$ contains a fourth power}.
					\item{}
						Case 2. There is at least one state $p_1$ that is visited at least 4 times
						and another state $p_2\ne p_1$ that is visited at least 3 times.
						Then by Theorem \ref{nfa9fact}, there is a cube $xxx$ and a square $yy$ in $w$.
						Since \textbf{$w$ has no squares of length $>4$}, we must have
						$\abs{xx}\le 4$ and $\abs{yy}\le 4$ and hence $1\le \abs{x}\le 2$ and $1\le \abs{y}\le 2$.
						We next consider possible lengths of $x$ and $y$.
						\begin{itemize}
							\item{}
								Subcase $\abs{x}= 2$. Say $x=ab$ where $\abs{a}=\abs{b}=1$.
								If $a\ne b$ then $xxx\in\{101010,010101\}$ so 1010 occurs in $w$, but \textbf{$w$ does not contain 1010};
								if $a=b$ then 0000 or 1111 occurs in $w$, contra assumption.
							\item{} Subcase $\abs{x}=1$, $\abs{y}= 2$:
								In this case, the $xxx$ and $yy$ occurrences must be disjoint,
								because the states in a $yy$ occurrence are $p_2 p_3p_2p_3p_2$ for some
								$p_3$ which must be disjoint from $p_1p_1p_1p_1$ when $p_1\ne p_2$.
								But then we can replace these by $p_2p_3p_2p_3p_2p_3p_2$ and $p_1p_1$, respectively,
								giving two distinct state sequences leading to acceptance, contradicting Lemma \ref{precisedef}.
							\item{}
								Subcase $\abs{x}=1$, $\abs{y}=1$:
								In this case again the occurrences of $xxx$ and $yy$ must be disjoint, since $p_1\ne p_2$.
								We can replace $p_1^4$ and $p_2^3$ by $p_1$ and $p_2^6$, respectively,
								again contradicting Lemma \ref{precisedef}.
						\end{itemize}
						\item{}
						Case 3. There are at least 3 states $p_1$, $p_2$, $p_3$ (all distinct) that are each visited at least 3 times.
						Then by Theorem \ref{nfa9fact}, there are three squares ${u_i}{u_i}$ at three distinct states $p_i$, $1\le i\le 3$.
						By assumption $\abs{{u_i}{u_i}}\le 4$ so $\abs{u_i}\le 2$.
						\begin{itemize}
							\item{}
								Subcase 3.1. $\abs{u_i}=\abs{u_j}=1$ for two values $1\le i< j\le 3$.
								Then the argument is entirely analogous to that in Case 2.
							\item{} Subcase 3.2
								$\abs{u_j}=\abs{u_k}=2$ for two values $1\le j < k\le 3$.
						\begin{itemize}
							\item{} Subsubcase 3.2.1.
								If disjoint, we can replace $u_j^2$ by $u_k^2$ to get $u_k^4$, again a \textbf{fourth power},
								by the argument of Subcase 3.1.
							\item{} Subsubcase 3.2.2.
								If nondisjoint with full overlap then
								\[
									{p_j}{a_1}{p_j}{a_2}{p_j}
								\]
								and
								\[
									{p_k}{b_1}{p_k}{b_2}{p_k}
								\]
								become
								\[
									{p_j}{p_k}{p_j}{p_k}{p_j}{p_k}
								\]
								and
								immediately we get \textbf{$10101$ or $01010$ or a fourth power in $w$};
							\item{} Subsubcase 3.2.3.
								If partial overlap only then ${p_j}{a_1}{p_j}{a_2}{p_j}$ and ${p_k}{b_1}{p_k}{b_2}{p_k}$ become,
								by Lemma \ref{precisedef},
								${p_j}a {p_j}a {p_j}$ and ${p_k}b {p_k}b {p_k}$ and then
								\[
									{p_j} a {p_j}{p_k}{p_j}{p_k} b {p_k}
								\]
								By Lemma \ref{precisedef} again, this must be
								\[
									{p_j}{p_k} {p_j}{p_k} {p_j}{p_k} {p_j}{p_k} = {({p_j}{p_k})}^4
								\]
								and so the read word must be of the form $abababa$,
								giving \textbf{an occurrence of $1010$ (if $a\ne b$) or of a 7th power (if $a=b$) in $w$}.
						\end{itemize}
					\end{itemize}
				\end{itemize}
				Thus all cases are covered and the Theorem is proved.
			\end{proof}

			\begin{cor}
				There is an infinite binary word having hereditary deficiency bound of 1.
			\end{cor}
			\begin{proof}
				We have two distinct proofs. On the one hand,
				Fraenkel and Simpson \cite{MR1309124} show there is an infinite almost {\squarefree} binary word,
				and the result follows from Theorem \ref{almostMain}.
				On the other hand, the infinite Thue--Morse word is overlap-free (Theorem \ref{thueStrong})
				and the result follows from Theorem \ref{main}.
			\end{proof}
			\begin{con}\label{evidence}
				There is an infinite binary word having hereditary deficiency 0.
			\end{con}
			\begin{rem}
				We obtained some numerical evidence for Conjecture \ref{evidence}.
				For instance, we found that there are 108 binary words of length 18 having hereditary deficiency 0.
			\end{rem}
	\def\cprime{$'$}

\end{document}